\documentclass[runningheads]{llncs}

\usepackage[utf8]{inputenc}

\usepackage{amsmath}
\usepackage{amssymb}
\usepackage{graphicx}
\usepackage{enumerate}
\usepackage{cite}
\usepackage{color}
\usepackage{xspace}
\usepackage{pdfsync}
\usepackage[loose]{subfigure}

\usepackage{algorithm}
\usepackage[noend]{algpseudocode}

\newcommand{\NP}{\ensuremath{\mathcal{NP}}}

\title{On 4-Map Graphs and 1-Planar Graphs and their Recognition Problem
\thanks{Supported by the Deutsche Forschungsgemeinschaft (DFG), grant
    Br835/18-1.}}
\titlerunning{4-Map Graphs and 1-Planar Graphs}
\author{Franz J. Brandenburg}
\authorrunning{Franz J. Brandenburg}
\institute{University of Passau, 94030 Passau, Germany \\
  \email{brandenb@informatik.uni-passau.de}}

\begin{document}

\maketitle

\begin{abstract}
We establish a one-to-one correspondence between  1-planar graphs
and general and hole-free 4-map graphs and show that 1-planar graphs
can be recognized in polynomial time  if they are
crossing-augmented, fully triangulated,
% planar-maximal,
and maximal 1-planar, respectively, with a polynomial of degree
$120, 3,$ and $5$, respectively.
\end{abstract}


\section{Introduction}
\label{sect:intro}

Planarity is one of the most basic and influential concepts in graph
theory. Many properties of planar graphs have been explored,
including embeddings, duality, and minors.  There are many linear
time algorithms for their  recognition as well as for the
construction of straight-line grid drawings, see
\cite{t-handbook-GD}.

There were several attempts to generalize planarity to ``beyond''
planar graphs. Such graph allow crossings of edges with
restrictions. (In other works the term near, nearly  or almost
planar is used). Such attempts are important, since many graphs are
not planar. A prominent example is 1-planar graphs, which were
introduced by Ringel \cite{ringel-65} in an approach to color a
planar graph and its dual. A graph is \emph{1-planar} if it can be
drawn in the plane such that each edge is crossed at most once.
These graphs have found recent interest, in particular in graph
drawing, as presented by Liotta \cite{l-beyond-14}. Special cases
are IC-planar
  and outer 1-planar
graphs. A graph is \emph{IC-planar}
\cite{ks-cpgIC-10,zl-spgic-13,bdeklm-IC-15} if it has an embedding
with at most one crossing per edge and in which each vertex is
incident to at most one crossing edge. If a graph can be embedded in
the plane with all vertices in the outer face and at most one
crossing per edge, then it is \emph{outer 1-planar}
\cite{abbghnr-ro1pglt-13, abbghnr-o1p-15, de-eogracd-12,
heklss-o1p-13, heklss-ltao1p-15}.

Beyond planarity may also be defined in terms of maps. A  \emph{map}
$\mathcal{M}$ is a partition of the sphere into finitely many
regions. Each region is a closed disk
% homeomorph
and the interiors of two regions are disjoint. Some regions are
labeled as \emph{countries}, and the remaining regions are lakes or
\emph{holes} of $\mathcal{M}$. In the plane, we use the region of
one country as the outer region, which is unbounded and encloses all
other regions. An \emph{adjacency} is defined by a touching of
regions. There is a \emph{strong} adjacency between two countries if
the boundaries of their regions intersect in a segment and  a
\emph{weak} adjacency if the boundaries intersect in a point. A map
$\mathcal{M}$ defines a graph $G$ such that the countries of
$\mathcal{M}$ are in one-to-one correspondence with the vertices of
$G$ and there is an edge if and only if the respective countries are
adjacent. Then $G$ is called a \emph{map graph}  and $\mathcal{M}$
the map of $G$.

Obviously, $k$ regions meeting at a point induce $K_k$ as a subgraph
of $G$. If no more than $k$ regions meet at a point, then
$\mathcal{M}$ is a $k$-map and $G$ a $k$-map graph. Map graphs were
introduced by Chen et al. \cite{cgp-pmg-98} and further studied in
\cite{cgp-mg-02, cgp-rh4mg-06}. Chen et al. observed that ordinary
planar graphs are the 2-map  or 3-map graphs and characterized map
graphs as half squares of bipartite planar graphs. Given a bipartite
graph $B = (V,U, M)$, its half square $H^2(B)$ is a graph with
vertices $V$ and whose edges are the paths of length two in $B$.
Chen et al. also proved that there are map graphs $G$ such that
$G-e$ is not a map graph.

In general, holes are necessary for the representation of graphs by
maps, since, e.g., grids cannot be represented, otherwise. If
$\mathcal{M}$ has no holes, then it is a \emph{hole-free} map and
its map graph $G$ is a \emph{hole-free map graph}. A hole-free map
looks like the dual of a planar graph. However, an adjacency at a
point includes weak adjacency. Chen et al. remark that a map graph
$G$ is hole-free if and only if the boundary of each face of the
bipartite graph $B$ with $G = H^2(B)$ has exactly four or six edges.
In \cite{cgp-rh4mg-06} they established a cubic time recognition
algorithm  for hole-free 4-map graphs. They also observed that the
3-connected hole-free map graphs are exactly the triangulated
1-planar graphs. A triangulated 1-planar graph has a 1-planar
embedding such that the boundary of each face consists of exactly
three edges or edge segments (up to a crossing point). We shall
extend this result and shall characterize 4-map graphs and hole-free
4-map graphs in terms of 1-planar graphs.

Given a (new) class of graphs, the recognition problem is always a
challenge. In general, the complexity is in the range between linear
time and \NP-completeness. Both extremes are reached by planar and
1-planar graphs. In general, 1-planarity is \NP-complete, as shown
by Grigoriev and Bodlaender \cite{GB-AGEFCE-07} and by Korzhik and
Mohar \cite{km-mo1ih-13}, and it remains \NP-complete even for
graphs of bounded bandwidth, pathwidth or treewidth
\cite{bce-pc1p-13}. \NP-completeness also holds for 3-connected
1-planar graphs with a given rotation system \cite{abgr-1prs-15},
i.e., the question whether a rotation system of a 3-connected graph
is 1-planar. Also, IC-planarity is \NP-complete, both for a graph
and a
 rotation system \cite{bdeklm-IC-15}.
In additon, deciding whether a planar graph is sub-Hamiltonian  is
\NP-complete \cite{w-chcppg-82}. A graph is sub-Hamiltonian if it is
a subgraph of a planar graph with a Hamilton circuit if and only if
it admits a two-page book embedding \cite{bk-btg-79}.

 Linear time algorithms for the recognition of planar
graphs have attracted many researchers, as Patrignani's survey in
\cite{p-pte-13} documents. Clearly, using any linear time algorithm
for planarity, it can be checked whether an embedding is 1-planar,
IC-planar and outer 1-planar, respectively. Independently and
simultaneously, Auer et al. \cite{abbghnr-ro1pglt-13} and Hong et
al. \cite{heklss-o1p-13} developed   linear time   algorithms for
the recognition of outer 1-planar graphs, see also
\cite{abbghnr-o1p-15, heklss-ltao1p-15}. Recently, Brandenburg
\cite{b-optimal1p-15} showed that optimal 1-planarity can be decided
in linear time. An optimal 1-planar graph has $4n-8$ edges.
Moreover, Eades et al. \cite{ehklss-tm1pg-13b} developed  a linear
time algorithm to test whether a rotation system is 1-planar. They
described their algorithm for maximal 1-planar graphs, but it also
goes through for  crossing augmentations, which are defined in
below.

 Surprisingly, map graphs  with holes are
feasible as shown by Thorup \cite{t-mgpt-98}. Chen et al.
\cite{cgp-rh4mg-06} remark   that Thorup's algorithm has a running
time of about $O(n^{120})$, and that it does not imply a polynomial
time recognition for $k$-map graphs and hole-free $k$-map graphs.
They detail a cubic time recognition algorithm $\mathcal{A}$ for
hole-free 4-map graphs.

In this paper, we   characterize 4-map graphs as crossing-augmented
1-planar graphs and hole-free 4-map graph as fully triangulated
1-planar graphs. The terms crossing-augmented and fully triangulated
are defined in Section \ref{main}. Then we use the recognition
algorithm of Chen et al. \cite{cgp-rh4mg-06} to show that fully
triangulated  and maximal 1-planar graphs can be recognized in
$O(n^3)$ and $O(n^5)$ time, respectively. Finally, we generalize the
test for a 1-planar rotation system of Eades et al.
\cite{ehklss-tm1pg-13b}  to crossing-augmented 1-planar graphs.

The paper is organized as follows. Section \ref{basics} describes
basic definitions. In Section \ref{main} we explore the relationship
between 1-planar graphs and map graphs and establish our results. We
conclude with some open problems in Section \ref{conclusion} and
given an answer on   conjectures of Chen et al. \cite{cgp-mg-02}.

\section{Foundations}
\label{basics}

We consider undirected graphs $G = (V, E)$ with $n$ vertices and $m$
edges. Unless otherwise stated, the graphs are simple and
2-connected.
An \emph{embedding} $\mathcal{E}(G)$   is a mapping of $G$ into the
plane or the sphere such that the vertices are mapped to distinct
points and each edge is a Jordan arc between its endpoints.
Crossings of incident edges with the same endpoint and
self-intersections are excluded. An embedding defines a
\emph{rotation system} $\mathcal{R}(G)$, which is a cyclic list of
incident edges or neighbors at each vertex. The embedding is
\emph{planar} if (the Jordan arcs of the) edges do not cross and
\emph{1-planar} if each edge is crossed at most once. We say that a
graph is \emph{planar} (\emph{1-planar}) if it has a planar
(1-planar, respectively) embedding, and accordingly for a rotation
system. The embedding $\mathcal{E}(G)$ is a witness for planarity
and 1-planarity, respectively, and it must satisfy the cyclic order
at each vertex in case of a given rotation system.

A planar embedding of a graph partitions the  plane (sphere)  into
\emph{faces}, which are closed disks (except for the outer face) and
are each specified by a cyclic sequence of edges (or the respective
vertices) that forms the boundary.
 In 1-planar
embeddings, a crossing subdivides an edge into two \emph{edge
segments}, and the  planarization  takes the crossing points and
edge segments into account and treats them as vertices and edges,
respectively.

\iffalse

 We consider two important  subclasses of 1-planar graphs.  A graph is
\emph{IC-planar} \cite{ks-cpgIC-10,zl-spgic-13,bdeklm-IC-15} if it
has an embedding with at most one crossing per edge such that each
vertex is incident to at most one crossing edge. If a graph can be
embedded in the plane with all vertices in the outer face and at
most one crossing per edge, then it is \emph{outer 1-planar}
\cite{abbghnr-ro1pglt-13, abbghnr-o1p-15, de-eogracd-12,
heklss-o1p-13, heklss-ltao1p-15}. IC-planar graphs were introduced
by Kr{\'{a}}l  and  Stacho \cite{ks-cpgIC-10}, who proved that they
are 5-colorable, and Brandenburg et al. \cite{bdeklm-IC-15} showed
that they are RAC graphs \cite{del-dgrac-11} and have an \NP-hard
recognition problem, even if the rotation system is fixed. Auer et
al. \cite{abbghnr-o1p-15} established that outer 1-planar graphs are
planar and extend outer planar graphs by one unit in parameters such
as treewidth, colorability, stack number and queue number, and there
are two independent linear time recognition algorithms by Auer et
al. \cite{abbghnr-ro1pglt-13} and Hong et al. \cite{heklss-o1p-13}.

\fi

 Given a class of graphs $\mathcal{G}$, a graph $G \in
\mathcal{G}$ is \emph{planar-maximal}, \emph{maximal} and
\emph{optimal}, respectively, if no further edge can be added to $G$
without inducing a crossing with some edge of $G$, violating the
defining property of $\mathcal{G}$, and violating the upper bound
for the number of edges of graphs in $\mathcal{G}$, respectively.
Hence, a graph  in $\mathcal{G}$ is maximal if there is no
supergraph in $\mathcal{G}$ with the same set of vertices and a
proper superset of edges, and it is optimal if there is no graph in
$\mathcal{G}$ of the same size and with more edges.
  Accordingly, an embedding $\mathcal{E}(G)$  of $G
\in \mathcal{G}$ is maximal (planar-maximal), if any edge added to
$\mathcal{E}(G)$ violates the defining properties of $\mathcal{G}$
(or is crossed, respectively). Thus, a graph $G$ is planar-maximal
(maximal)  if every embedding satisfying the properties of
$\mathcal{G}$ is planar-maximal (maximal, respectively). We call a
graph $G$ \emph{plane-maximal} 1-planar if  $G$ has a planar-maximal
embedding. In fact, every triangulated 1-planar graph is
plane-maximal but not necessarily planar-maximal 1-planar.
Note the difference between planar-maximal embeddings and graphs. As
an example, consider $K_5-e$, which is a maximal planar graph and
whose planar embedding is planar-maximal 1-planar. However, the
removed edge $e$ can be added and drawn planar if a $K_4$ subgraph
of $K_5-e$ is drawn   with a pair of crossing edges. Hence, $K_5-e$
is plane-maximal and not planar-maximal  1-planar or IC-planar.

Clearly, the concepts maximal-planar, maximal and optimal coincide
on planar graphs, and the maximum number of edges is $3n-6$.
Bodendiek et al. \cite{bsw-bs-83} showed that optimal 1-planar
graphs have $4n-8$ edges and that such graphs exist for $n=8$ and
all $n \geq 10$ \cite{bsw-1og-84}. The upper bound  was rediscovered
in many works. Bodendiek et al. also observed that there are maximal
1-planar graphs, which are not optimal. The gap in the number of
edges of maximal 1-planar is quite large, as shown by Brandenburg et
al. \cite{begghr-odm1p-13}, who found sparse maximal 1-planar graphs
with only $\frac{45}{17}n - \frac{84}{17}$ many edges.
Similarly, optimal IC-planar graphs have at most $\frac{13}{4}n - 6$
edges \cite{zl-spgic-13} and there are optimal IC-planar achieving
this bound. However,  there are maximal IC-planar graphs with only
$3n-5$ edges. For the latter, consider graphs   as displayed in Fig.
\ref{fig:sparseIC} and note that  maximal IC-planar graphs are
supergraphs of maximal planar graphs.
 Auer et al. \cite{abbghnr-o1p-15}
 observed that outer 1-planar
graphs have at most $2.5n-4$ edges, whereas there are maximal outer
1-planar graphs with $\frac{11}{5}n - \frac{18}{5}$ many edges, and
both bounds are tight.

\begin{figure}
   \begin{center}
     \includegraphics[scale=0.6]{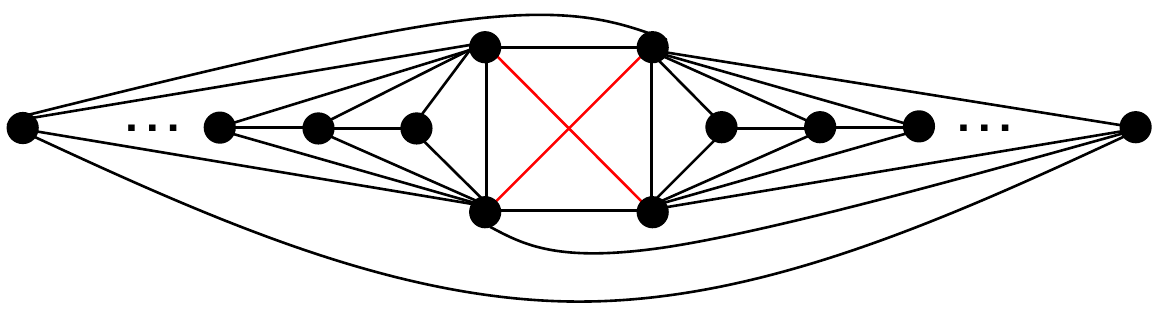}
     \caption{Sparse maximal IC-planar graphs with $3n-5$ edges
     \label{fig:sparseIC}}
   \end{center}
\end{figure}

The complete graph  on four vertices $K_4$ and its embedding plays a
crucial role. It can be drawn planar as a \emph{tetrahedron}  and
with a pair of crossing edges as a \emph{kite}, see Fig.
\ref{K4layout}.   In fact, there are four embeddings as a kite: fix
$a$ and flip $b,c$ and $c,d$ \cite{Kyncl-09}. In the terminology of
Chen et al. \cite{cgp-mg-02, cgp-rh4mg-06} on map graphs, a
tetrahedron corresponds to a rice-ball and a kite to a  pizza, see
Fig. \ref{K4maps}.

\begin{figure}
  \centering
  \subfigure[] {
    \includegraphics[scale=0.40]{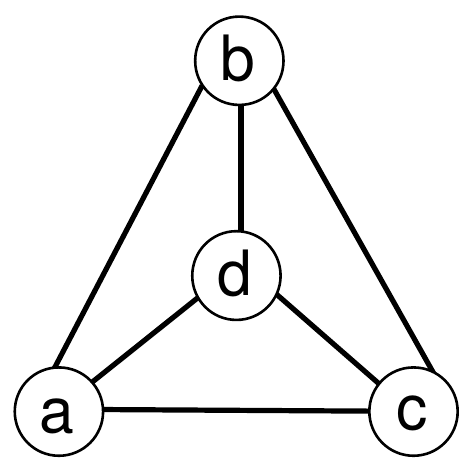}
    \label{fig:tetrahedron}
  } \quad\quad\quad\quad
  \subfigure[] {
      \includegraphics[scale=0.40]{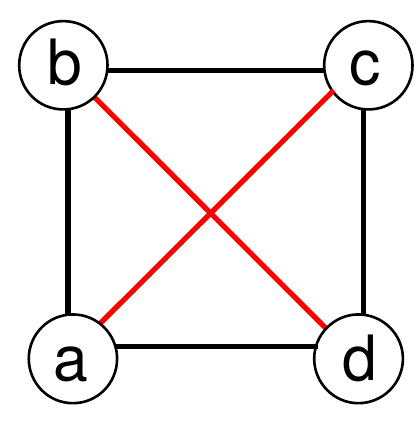}
      \label{fig:kite}
  }
   \caption{Drawings of $K_4$  \subref{fig:tetrahedron} planar as a tetrahedron and \subref{fig:kite} with a crossing
   as a kite.}
  \label{K4layout}
\end{figure}

\begin{figure}
   \centering
  \subfigure[] {
    \includegraphics[scale=0.40]{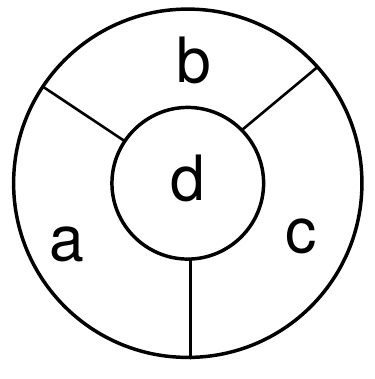}
    \label{fig:riceball}
  } \quad\quad\quad\quad
  \subfigure[] {
      \includegraphics[scale=0.40]{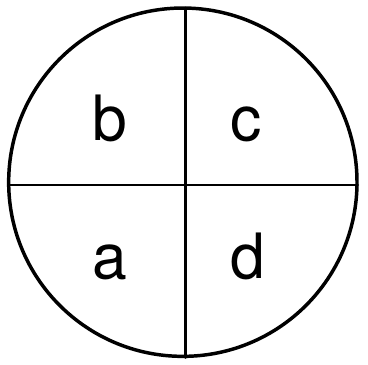}
      \label{fig:pizza}
  }
   \caption{A map of $K_4$ \subref{fig:riceball} as a rice-ball  and \subref{fig:pizza} as a  pizza .}
  \label{K4maps}
\end{figure}

We summarize the bounds on the complexity of subclasses of 1-planar
graphs in Table \ref{table1}.

\begin{table}
\begin{center}
\begin{tabular}{l |c | c | c}
                    \, & 1-planar & IC-planar & \, outer 1-planar \\
         \hline
                    graph  & \, \NP-complete \cite{GB-AGEFCE-07, km-mo1ih-13}\,& \, \NP-complete
                    \cite{bdeklm-IC-15} \, & $O(n)$ \cite{abbghnr-o1p-15, heklss-ltao1p-15}  \\
                    rotation system   &    \NP-complete \cite{abgr-1prs-15} &  \NP-complete
                    \cite{bdeklm-IC-15} & $O(n)$ \\
                    crossing-augmented \, & $O(n^{120})$ \cite{t-mgpt-98, cgp-rh4mg-06}&   ? & $O(n)$ \\
                    fully triangulated \, & $O(n^3)$ & ? &  $O(n)$ \\
                    plane-maximal \, & ? & ? & $O(n)$ \\
                    planar-maximal   &  ? &   ? & $O(n)$ \cite{abbghnr-o1p-15} \\
                    maximal &  $O(n^5)$ &   ? & $O(n)$ \cite{abbghnr-o1p-15} \\
                    optimal &  $O(n)$ \cite{b-optimal1p-15}   &   ? & $O(n)$   \\
\end{tabular}
\end{center}
\caption{The recognition complexity of 1-planar  graphs}
\label{table1}
\end{table}

\section{Polynomial time solvable instances}
\label{main}

In this section we characterize 4-map and hole-free 4-map graphs in
terms of 1-planar graphs and show that maximal-planar and maximal
1-planarity can be recognized  in $O(n^3)$ and $O(n^5)$ time,
respectively.

It was first observed by Ringel \cite{ringel-65} and   rediscovered
many times that a pair of crossing edges  in a 1-planar embedding
can be augmented  to form a kite. This augmentation seems to make
the difference between tractable and intractable instances of
1-planar graphs. However, augmentation needs an embedding.

\begin{definition}
A 1-planar embedding $\mathcal{E}(G)$ of  a 1-planar graph $G$ is
\emph{crossing-augmented} if for every pair of crossing edges
$\{a,b\}$ and $\{c,d\}$ in $\mathcal{E}(G)$  there are the edges
$\{a,b\}$, $\{b,c\}$, $\{c,d\}$, $\{d,a\}$ in $G$.
\end{definition}

Obviously, planar-maximal and maximal implies crossing-augmented.
More importantly, we can improve upon the observation of Chen et al.
\cite{cgp-rh4mg-06} that the triangulated 1-planar graphs are
exactly the 3-connected  hole-free  4-map graphs. Triangulated means
that the boundary of each face in an embedding consists of  exactly
three edges or edge segments (up to a crossing point), and it
enforces 3-connectivity.

\begin{theorem} \label{kite-augmented}
A 1-planar graph is crossing-augmented if and only if it is a 4-map
graph (with holes).
\end{theorem}
\begin{proof}
A graph $G = (V,E)$ is a 4-map graph if and only if $G = H^2(B)$ for
a planar bipartite graph $B = (V,U, F)$ \cite{cgp-mg-02} with
vertices of degree two  and four  in $U$. Construct a 1-planar
embedding of $G$ from the embedding of $B$ by contracting the edges
incident to degree-2 vertices of $U$ and replacing each vertex of
degree four of $U$ by a kite. Then $G$  is crossing-augmentation.
Conversely, an embedded planar bipartite graph $B$ is obtained from
$\mathcal{E}(G)$ by subdividing each planar edge and replacing each
crossing point by a degree-4 vertex of $U$ such that $G = H^2(B)$.
 \qed
\end{proof}

The  result of Thorup \cite{t-mgpt-98} and the remark of Chen et al.
\cite{cgp-rh4mg-06}  is  used for our first  recognition problem on
1-planar  graphs.

\begin{corollary} \label{thm-kite}
For a graph  $G$, it takes polynomial time (of degree about $120$)
to test whether $G$ is crossing-augmented 1-planar.
\end{corollary}

When considering  maps as  dual graphs, one must avoid vertices of
degree greater than four in the dual graphs if 4-maps and 4-map
graphs are taken into account. Hence, the faces of an embedded graph
should be triangles or look like kites. However, this presupposition
is not granted if the given graph is not 3-connected.
If there is a separation pair $\{u,v\}$ and $G-\{u,v\}$ decomposes
into components $G_1, \ldots, G_k$ for some $k>1$, then $G$ is
1-planar if and only if each of the graphs $G_i + e$ is 1-planar
with $e=(u,v)$ as a planar edge, as noted by Chen et al.
\cite{cgp-rh4mg-06}. Similarly, Brandenburg \cite{b-vr1pg-14} has
introduced copies of the edge $(u,v)$ to separate the components at
a separation pair. This idea must be retained.

\begin{definition}
A 1-planar embedding $\mathcal{E}(G)$ is \emph{fully triangulated}
if the separated embedding $\mathcal{E}_s(G)$ is triangulated, i.e.,
the boundary of each face consists of three edges or edge segments
up to a crossing point. $\mathcal{E}_s(G)$ is obtained from
$\mathcal{E}(G)$ by adding a copy of the edge $(u, v)$ to separate
components   at a separation pair $\{u,v\}$. Accordingly, a graph is
\emph{fully triangulated 1-planar} if it admits such an embedding.
\end{definition}

Chen et al. \cite{cgp-rh4mg-06} noted that a graph $G$ is a 4-map
graph if and only if $G$ is a triangulated 1-planar graph, provided
$G$ is 3-connected. We generalize this result to 2-connected graphs.
Note that maps as well as planar dual graphs enforce 2-connectivity.

\begin{theorem} \label{kite-augmented}
A 1-planar graph is fully triangulated if and only if it is a
hole-free 4-map graph.
\end{theorem}
\begin{proof}
Let $\mathcal{E}_s(G)$ be a weakly triangulated embedding. Then
remove one edge from each pair of crossing edges. The resulting
graph $H$ is triangulated including  copies of the edge between
separation pairs. The hole-free 4-map is obtained from an embedding
of the dual $H^*$ by contracting the edge between the vertices of
triangular faces, which were obtained by a removal of a crossing
edge.  Each contraction merges the end-vertices and results in a
4-point. The weak adjacency returns the formerly removed crossing
edge. Note that $H^*$ may have multiple edges.

Conversely, if $\mathcal{M}$ is a hole-free 4-map of $G$, then take
$\mathcal{M}$  as a planar dual $H^*$ and construct the planar
primal graph $H$. $H$ has multiple edges between vertices $u$ and
$v$ if and only if the boundary between two regions is not a simple
curve or path if and only if $\{u,v\}$ is a separation pair of $G$.
At each 4-point of $\mathcal{M}$ with countries $a,b,c,d$ in this
order, $H$ includes the edges $(a,b), (b,c), (c,d), (d,a)$. Add the
edges $(a,c)$ and $(b,d)$ to  obtain a graph  $G^{\prime}$ such that
there is a kite with vertices $a,b,c,d$ in the embedding
$\mathcal{E}(G^{\prime})$, which is  obtained from the embedding of
$H$. Here, multiple copies of edges are taken into account. Finally,
remove all but one copy of each multi-edge form $G^{\prime}$ such
that the resulting graph is simple. This graph is $G$.
 \qed
\end{proof}

Crossing-augmentation and triangulation enforce distinctions between
1-planar and map graphs. Chen et. al \cite{cgp-mg-02} have shown
that the removal of an edge destroys map graphs. Their example can
be used to show that map graphs are not closed under subdivision. On
the other hand, 1-planar graphs are closed under taking subgraphs
and subdivisions.  In fact, every graph can be obtained from a
1-planar graph by subdivisions. Hence, neither map graphs nor
1-planar graphs  can be characterized by minors.

Using the the cubic time recognition algorithm of Chen et al.
\cite{cgp-rh4mg-06} and the SPQR-decomposition for the detection of
all separation pairs \cite{dt-olpt-96}, we immediately obtain:

\begin{corollary}
For a graph  $G$, it takes   $\mathcal{O}(n^3)$ time to test whether
 $G$ is fully triangulated 1-planar.
\end{corollary}

\begin{theorem}
For a graph  $G$, it takes   $\mathcal{O}(n^5)$ time to test whether
 $G$ is maximal 1-planar.
\end{theorem}

\begin{proof}
Clearly, a graph $G$ is maximal 1-planar if $G$ is maximal 1-planar
and $G+e$ is not for any new edge $e$ added to $G$, and each of the
$O(n^2)$ tests takes $O(n^3)$ time.
 \qed
\end{proof}

We would like to establish tractability also for plane-maximal and
planar-maximal 1-planar graphs.  The obstacle is the variety of
1-planar embeddings. There are even optimal 1-planar graphs with
different embeddings, see \cite{s-s1pg-86, s-rm1pg-10}. Algorithm
$\mathcal{A}$ of Chen et al. \cite{cgp-rh4mg-06} embeds a $K_4$ as a
kite (correct pizza), whenever possible, and then ``makes progress''
by removing one crossing edge. However, there are places, such as a
so-called separating edge, where $\mathcal{A}$ has a choice. If
$\mathcal{A}$ computes a planar-maximal embedding of a graph $G$,
then there may be another embedding such that a planar edge can be
added. Conversely, if the computed embedding is not planar-maximal,
there may be a planar-maximal one. However, we can only reduce the
general case to the 3-connected case.

\begin{definition}
A 1-planar embedding $\mathcal{E}(G)$ with a planar edge $(a,b)$ in
the outer face is called \emph{open} if after the removal of $(a,b)$
there is a vertex $v$ of $G$ in the outer face. $v$ is called
\emph{open vertex}. Otherwise, $G$ is called \emph{closed}. A
1-planar graph with a distinguished edge $(a,b)$ is open if it has
an open embedding.
\end{definition}

A  1-planar graph is closed if its   embedding is a W-configuration
of Thomassen \cite{t-rdg-88}, see Fig. \ref{fig:W-config}.
W-configurations do not allow straight-line 1-planar drawings, as
noted in \cite{t-rdg-88} and \cite{help-ft1pg-12}. An embedding is
open at one or two sides. In the first case it is a B-configuration
\cite{t-rdg-88},  see \ref{fig:B-config} and it has a planar
interface  if it is two-sided open.

\begin{figure}
  \centering
  \subfigure[] {
    \rotatebox{90}{\includegraphics[scale=0.60]{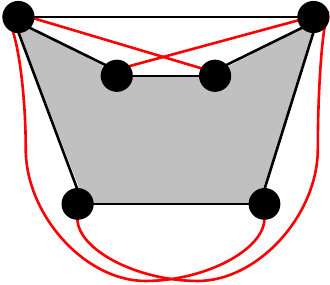}}
    \label{fig:W-config}
  }\quad\quad\quad \quad
  \subfigure[] {
      \rotatebox{90}{\includegraphics[scale=0.6]{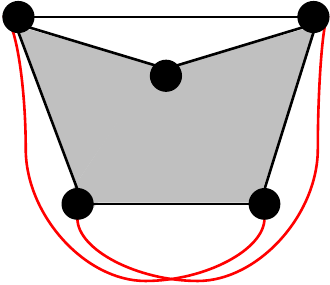}}
      \label{fig:B-config}
 }
   \caption{A \subref{fig:W-config} W-configuration and a \subref{fig:B-config} B-configuration.
   There may be subgraphs in the outer face and in the shaded area
   }
  \label{configs}
\end{figure}

In a map, the boundary between two regions $u$ and $v$ is not a
simple curve and looks like a chain of pearls, see Fig.
\ref{fig:pearls}. Each pearl represents a component $G_i$ of $G-\{u,
v \}$ at a separation pair $\{u, v \}$ and has  a left and a right
contact point. Now, $G_i$ is closed if and only if both contact
points   are 4-points. If both contact points are 3-points, then
$G_i$ is open and planar, and it is a B-configuration, if  one
contact point is a 4-point.

\begin{figure}
   \begin{center}
     \includegraphics[scale=0.33]{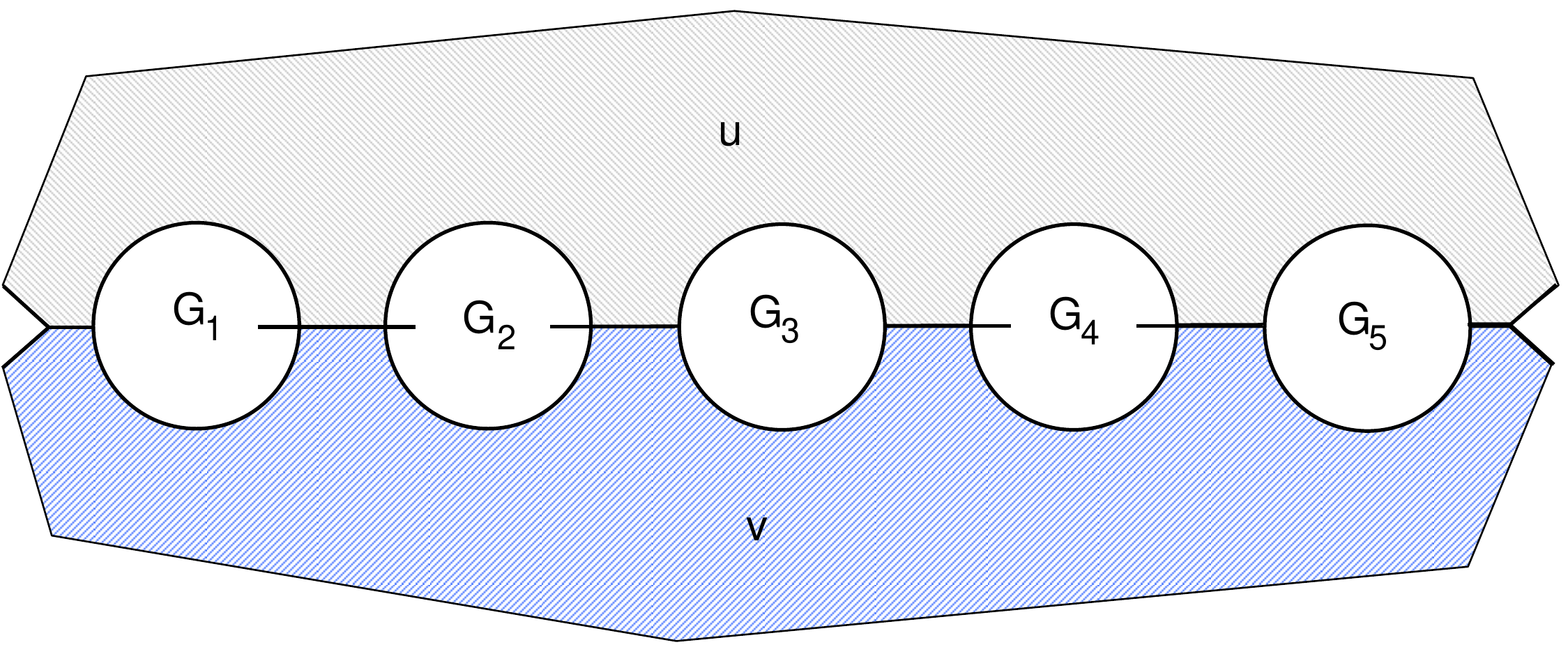}
     \caption{A map of a separation pair $\{u,v\}$ with a chain of pearls.
     $G_1$ is open with a B-configuration, $G_2$ and $G_4$ are closed, and
     $G_3$ and $G_5$ are open and planar.
     \label{fig:pearls}}
   \end{center}
\end{figure}

There are many embeddings of the components at a separation pair.
Each component can be flipped and the components can arbitrarily be
permuted. This corresponds to a flip of the pearls and their
permutation in a map.

\begin{lemma}
There is a linear-time reduction from the problem of deciding
whether a given 1-planar graph is planar-maximal and plane-maximal,
respectively, to the special case where the graph is 3-connected.
\end{lemma}
\begin{proof}
Clearly, a graph $G$ is (plane or planar-maximal) 1-planar if and
only if at every separation pair $\{u,v\}$ the components $G_i +e$
are (plane or planar-maximal) 1-planar, respectively, where $G -
\{u,v\}$ is decomposed into components  $G_1, \ldots, G_k$ for some
$k>1$ and the edge $e = (u, v)$ is planar. Each $G_i +e$ is a
subgraph of $G$ and thus remains 1-planar.

 Suppose that each $G_i + e$  is planar-maximal 1-planar.
 Then $G$ is planar-maximal if and only if at most one $G_i +
 \{u,v\}$ is open.
Similarly, if each $G_i + e$  is plane-maximal 1-planar, then $G$ is
plane-maximal if and only if the number of two-sided open components
does not exceed  the number of components with a closed embedding.
Then the components can be arranged such that two open vertices do
not appear in the same face.

 These properties are checked in linear time along the SPQR
decomposition tree, in which the input graph is recursively
decomposed at its separation pairs and at an edge if the component
is 3-connected, see \cite{dt-olpt-96}. \qed
\end{proof}

The parallel results for IC-planar graphs are not yet clear. If
algorithm  $\mathcal{A}$ finds a 1-planar embedding which is not IC,
then there may be another IC-planar embedding.

For outer 1-planar graphs, Auer et al. \cite{abbghnr-o1p-15} showed
that the recognition of  maximal-planar, maximal and outer 1-planar
graphs, respectively, can be
 solved in linear time. They use the decomposition of a graph into
its 2-connected components and retrieve planar-maximality and
maximality  directly from the structure of the SPQR-tree
\cite{dt-olpt-96}. For optimal outer 1-planarity one can either
check that the given graph is outer 1-planar and has $2.5n-4$ edges
or that the SPQR-tree is composed of kites. Similarly, properties
like plane-maximal, crossing-augmented and fully triangulated can
directly be recognized at the SPQR-tree.

\begin{theorem}
For a graph $G$, the following problems can be solved in linear
time.
\begin{enumerate}
\item  Is $G$ outer 1-planar  \cite{abbghnr-o1p-15, heklss-ltao1p-15}?
\item   Is  $G$ crossing-augmented outer 1-planar?
\item   Is  $G$ weakly triangulated outer 1-planar?
\item  Is $G$ plane-maximal ?
\item  Is $G$ planar-maximal  \cite{abbghnr-o1p-15}?
\item  Is $G$ maximal outer 1-planar \cite{abbghnr-o1p-15}?
\item  Is $G$ optimal outer 1-planar?
\end{enumerate}
\end{theorem}

Finally, we can improve  upon a result  of Eades et al.
\cite{ehklss-tm1pg-13b} on the recognition of 1-planar rotation
systems. The algorithm of Eades et al. considers walks  around a
face and finds a simple cycle if the face is planar and traverses
the crossing edges twice in opposite directions  before the walk
revisits an edge in the same direction if there is a kite. Simply
speaking, it uses the crossing edges as a bridge.

\begin{corollary}
There is a linear time algorithm to test whether a rotation system
is 1-planar if the underlying embedding is  crossing-augmented.
\end{corollary}

\section{Conclusion and Open Problems}
\label{conclusion}

We showed that 1-planarity  can be tested in polynomial time if the
graphs are crossing-augmented, planar-maximal, maximal and optimal,
respectively. \\
(i) Do similar results  also hold for IC-planarity?

There are many other classes of  beyond  planar graphs, such as
fan-planar \cite{bcghk-rfpg-14, bddmpst-fan-15}, bar 1-visibility
\cite{DEGLST-bkvg-07} and bar (1,j)-visibility graphs
\cite{bhkn-bvg-15}, right angle crossing graphs (RAC)
\cite{del-dgrac-11}, quasi-planar graphs \cite{aapps-qpg-97}, and
rectangle visibility graphs \cite{hsv-rstg-99}. \\
(ii) It is unknown
whether planar-maximality, maximality and optimality in these
classes can be recognized in polynomial time. For outer-fan planar
graphs, maximality can be tested in linear time
\cite{bcghk-rfpg-14}.

In general, 1-planar embeddings are not unique. However, it seems
that such embeddings are weakly equivalent if the graphs are
(planar-) maximal or optimal.   Two embeddings $\mathcal{E}_1(G)$
and $\mathcal{E}_2(G)$ of a graph $G$ are weakly equivalent if there
is a graph automorphism $\sigma : G \rightarrow G$ such that
$\mathcal{E}_1(G)$ is (topologically) equivalent   to
$\mathcal{E}_2(\sigma(G))$. Schumacher \cite{s-s1pg-86} and Suzuki
\cite{s-rm1pg-10} proved that optimal 1-planar graphs are weakly
equivalent. \\
(iii) Do maximal (planar-maximal) 1-planar and IC-planar graphs,
respectively,  have a unique embedding up to weak isomorphism?

Chen et al. \cite{cgp-mg-02} address a generalization of maps and
allow  a region $u$ to include another region. They conjecture that
the recognition problem for this generalization remains polynomially
solvable, which clearly holds true, since the enclosing region is an
articulation vertex of the map graph. Another generalization of Chen
et al. is unclear. The relation between two regions shall be left
unspecified. The regions may touch or not. If regions may touch, but
the respective vertices in the map graph are not connected by an
edge, then the resulting map graphs are subgraphs of 3-connected
1-planar graphs. For such graphs, the conjecture of Chen et al.
holds true, since the recognition problem for 1-planar graphs is
\NP-complete  \cite{abgr-1prs-15}. Note that there are 1-planar
graphs, such as the sparse maximal 1-planar graphs of Brandenburg et
al. \cite{begghr-odm1p-13}, which are not a subgraph of a
3-connected 1-planar graph.

\section{Acknowledgements}
  I would like to thank Christian Bachmaier for many inspiring discussions
  and his support, and Giuseppe Liotta for the hint on map graphs.

\bibliographystyle{splncs03}
\bibliography{brandybibV2}

\end{document}